\documentclass[letterpaper,10pt,twocolumn,twoside,journal]{IEEEtran}
\IEEEoverridecommandlockouts                         
\usepackage{times}
\usepackage{setspace}
\usepackage{url}
\spacing{1}
\usepackage[utf8]{inputenc}
\usepackage[T1]{fontenc}
\usepackage{graphicx}		
\usepackage{wrapfig}
\usepackage[format=plain,font=footnotesize,labelfont=bf,labelsep=period]{caption}
\usepackage{sidecap} 
\usepackage{subfig}
\usepackage[export]{adjustbox}
\usepackage[font=small]{caption}
\usepackage{float}
\usepackage{dblfloatfix}

\usepackage{amsmath} 
\usepackage{amssymb}  
\usepackage{amsthm}
\usepackage{mathtools}
\usepackage[normalem]{ulem}
\usepackage{paralist}	
\usepackage[space]{grffile} 
\usepackage{color}
\usepackage{bbm}
\usepackage{placeins} 
\usepackage{array}
\usepackage{siunitx}
\usepackage{hyperref}
\usepackage{enumitem}

\newtheorem{theorem}{Theorem}

\newtheorem{lemma}{Lemma}
\theoremstyle{definition}
\newtheorem{definition}{Definition}
\theoremstyle{remark}
\newtheorem{remark}{Remark}
\theoremstyle{definition}

\theoremstyle{definition}

\newcommand{\N}{\mathbb{N}}

\newcommand{\R}{\mathbb{R}}

\newcommand{\K}{\mathcal{K}}

\definecolor{darkblue}{RGB}{0,0,102}
\definecolor{lightblue}{RGB}{77,77,148}

\definecolor{gold}{RGB}{234, 170, 0}
\definecolor{metallic_gold}{RGB}{139, 111, 78}

\renewcommand{\cal}[1]{\mathcal{ #1 }}
\newcommand{\mb}[1]{\mathbf{ #1 }}
\newcommand{\bs}[1]{\boldsymbol{ #1 }}
\newcommand{\bb}[1]{\mathbb{ #1 }}

\DeclareMathOperator*{\argmin}{argmin}

\renewcommand{\baselinestretch}{.985}
\allowdisplaybreaks

\usepackage[ruled, vlined]{algorithm2e}
\usepackage{mathrsfs}
\usepackage{hyperref}

\newcommand{\Z}{\bb{Z}}
\theoremstyle{definition}
\newtheorem{claim}{Claim}

\title{\LARGE \textbf{Sampled-Data Stabilization with Control Lyapunov Functions \\ via Quadratically Constrained Quadratic Programs}}
\author{Andrew J. Taylor$^{1}$, Victor D. Dorobantu$^{1}$, Yisong Yue, Paulo Tabuada, and Aaron D. Ames
\thanks{$^1$Both authors contributed equally. A.J. Taylor, V.D. Dorobantu, Y. Yue, and A.D. Ames are with the Department of Computing and Mathematical Sciences, California Institute of Technology, Pasadena, CA 91125, USA, {\tt\small \{ajtaylor, vdoroban, yyue, ames\}@caltech.edu}. P. Tabuada is with the Department of Electrical Engineering, University of California at Los Angeles, Los Angeles, CA 90095, USA, {\tt \small tabuada@ucla.edu}.
}
}

\begin{document}

\maketitle
\thispagestyle{empty}
\begin{abstract}
Controller design for nonlinear systems with Control Lyapunov Function (CLF) based quadratic programs has recently been successfully applied to a diverse set of difficult control tasks. These existing formulations do not address the gap between design with continuous time models and the discrete time sampled implementation of the resulting controllers, often leading to poor performance on hardware platforms. We propose an approach to close this gap by synthesizing sampled-data counterparts to these CLF-based controllers, specified as quadratically constrained quadratic programs (QCQPs). Assuming feedback linearizability and stable zero-dynamics of a system's continuous time model, we derive practical stability guarantees for the resulting sampled-data system. We demonstrate improved performance of the proposed approach over continuous time counterparts in simulation. 
\end{abstract}

\section{Introduction}
\label{sec:intro}
Nonlinear control methods offer promising solutions to many modern engineering applications. However, theoretically sound controller designs often fail to achieve desired behaviors when deployed on real systems. Thus, it is critical to understand the discrepancies between theoretical design and practical implementation mathematically, and to design controllers that close these gaps. Specifically, we address the challenges in designing controllers with continuous time models and realizing them with discrete time sampling implementations. 

Feedback linearization is a powerful tool in nonlinear control design, enabling the algorithmic synthesis of controllers for a wide class of mechanical and electrical systems \cite{isidori2013nonlinear}. Moreover, feedback linearization provides a constructive method to find Control Lyapunov Functions (CLFs) \cite{freeman2008robust} for continuous time systems. This fact has been used to formulate stabilizing controllers through Quadratic Programs (QPs) \cite{ames2014rapidly, galloway2015torque}, seeing use in several applications such as robotics \cite{choi2020reinforcement} and autonomous vehicles \cite{xiao2019control}. Despite these successes, translating these controllers to hardware platforms often requires additional effort to overcome the degradation of performance and introduction of chatter caused by sample frequency limitations. 

We propose an extension of the preceding nonlinear controller designs to the \textit{sampled-data} setting \cite{monaco2007advanced}, in which control inputs are specified at discrete sample times and held constant between sequential sample times (referred to as a zero-order hold). The resulting evolution of such systems between sample times is described by discrete time models, for which exact representations can rarely be derived, motivating the synthesis of controllers with approximate discrete time models. The foundational work in \cite{nevsic1999sufficient, nesic2004framework} established a sampled-data framework for translating stability guarantees for an approximate discrete time model to the exact discrete time model. Resulting sampled-data synthesis methods \cite{nesic2001backstepping, nevsic2005lyapunov, grune2003optimization} specify controllers that often demonstrate improved performance over continuous time designs \cite{laila2003changing}, many using a simple Euler approximate discrete time model \cite{mareels1992controlling,nesic2001backstepping}, but optimization-based controllers synthesized using CLFs found via feedback linearization have not yet been considered.

The relationship between feedback linearizability and sampling has been thoroughly investigated \cite{grizzle1986feedback, monaco1986linearizing, grizzle1988feedback, arapostathis1989effect}. Much of this investigation has focused on whether feedback linearizability of a system's continuous time dynamics implies feedback linearizability of the exact discrete time model of the system (a fact that requires strict structure of the continuous time dynamics). Even with Euler approximate discrete time models, a continuous time feedback linearizable system must be first expressed in appropriate coordinates before sampling and approximating to ensure the approximate model is feedback linearizable in the discrete time sense \cite{grizzle1988feedback}. This requirement is also seen with higher-order approximate models obtained via Taylor expansion \cite{barbot1996sampled}. The work \cite{monaco1988zero, yuz2005sampled, nishi2009nonlinear} studies the zero-dynamics that arise due to sampling and higher-order approximations, but not the impact of sampling on the stability of existing continuous time zero-dynamics.

We make two main contributions in this work. First, we formally integrate feedback linearization and zero-dynamics with Euler approximate discrete time models for sampled-data systems via the results in \cite{nevsic1999sufficient}. In particular, we demonstrate that systems with feedback linearizable continuous time dynamics and locally exponentially stable zero-dynamics can be rendered practically-stable via a continuous time feedback linearizing controller when the inputs to the system are implemented with a zero order hold. The often local nature of stability of zero-dynamics requires modification of the global results in \cite{nevsic1999sufficient}. Second, we extend the preceding result to optimization-based controllers using CLFs \cite{galloway2015torque} synthesized via feedback linearization. In Section \ref{sec:sim} we propose a controller, specified via a convex, quadratically constrained quadratic program (QCQP), that replaces the standard affine constraint on the time derivative of the CLF with a quadratic constraint on the decrease of the CLF over a sample period (as approximated by the Euler discrete time model). We demonstrate the improved performance of this controller over continuous time CLF formulations with sample frequency limitations.

\section{Preliminaries}
\label{sec:background}
Throughout this work, we will consider the nonlinear control system governed by the differential equation:
\begin{equation}\label{eqn:nonlinear-dynamics}
    \dot{\mb{x}} = \mb{f}(\mb{x}) + \mb{g}(\mb{x})\mb{u},
\end{equation}
for state signal $\mb{x}$ and control input signal $\mb{u}$ taking values in $\R^n$ and $\R^m$, respectively, drift dynamics $\mb{f}: \R^n \to \R^n$, and actuation matrix function $\mb{g}: \R^n \to \R^{n \times m}$. Consider an open subset $\mathcal{Z} \subseteq \R^n \times \R^m$ and its projection onto the state space $\mathcal{X} \triangleq \pi_1(\mathcal{Z})\subseteq\R^n$. Assume there exists $T_{\mathrm{max}} \in \R_{++}$ such that for every state-input pair $(\mb{x}_0, \mb{u}_0) \in \mathcal{Z}$, there exists a unique solution $\bs{\varphi}:[0,T_{\mathrm{max}}]\to\R^n$ satisfying:
\begin{align}
\label{eqn:varphiode}
    \dot{\bs{\varphi}}(t) &= \mb{f}(\bs{\varphi}(t)) + \mb{g}(\bs{\varphi}(t))\mb{u}_0 \quad \forall t \in (0, T_{\max}),\\
    \bs{\varphi}(0) &= \mb{x}_0. \label{eqn:varphiic}
\end{align}
This enables the following reachable set definition:
\begin{equation*}
    \mathcal{D} \triangleq \big\{\mb{x}\in\R^n ~|~ \exists(\mb{x}_0,\mb{u}_0)\in\mathcal{Z},~ t\in[0,T_{\mathrm{max}}] ~\mathrm{s.t.}~ \mb{x} = \bs{\varphi}(t) \big\},
\end{equation*}
where $\mathcal{X}\subseteq\mathcal{D}$. Given an $h\in(0,T_{\mathrm{max}}]$, we define a controller $\mb{k}: \mathcal{X} \to \R^m$ as \textit{$h$-admissible} if for any state $\mb{x}_0 \in \mathcal{X}$, the state-input pair $(\mb{x}_0, \mb{k}(\mb{x}_0))$ satisfies $(\mb{x}_0, \mb{k}(\mb{x}_0)) \in \mathcal{Z}$ and the corresponding solution $\bs{\varphi}$ satisfies $\bs{\varphi}(t) \in \mathcal{X}$ for all $t\in[0,h]$. 
\begin{remark}
This requirement on $h$-admissible controllers ensures that in the sampled-data context, the evolution of the system may be described by iterative solutions to \eqref{eqn:varphiode}-\eqref{eqn:varphiic}. For many systems, the assumption that an $h$-admissible controller keeps the system's state in the set $\mathcal{X}$ is relatively weak as $\mathcal{X}$ is defined to ensure the continued existence of solutions rather than reflecting a task-specific set that must be kept invariant. In many cases verifying $h$-admissibility of a controller may be intractable, but the assumption of $h$-admissibility of a given controller is often weak in practice.
\end{remark}

Feedback linearization offers a tool for the synthesis of stabilizing controllers for nonlinear continuous time systems, and will serve an important role in constructing optimization-based controllers for sampled-data nonlinear systems. We make use of the following abbreviated definition, but more details may be found in \cite{isidori2013nonlinear}:

\begin{definition}[\textit{Feedback Linearizability}]
The system \eqref{eqn:nonlinear-dynamics} is \textit{feedback linearizable} if there exist dimensions $k, \gamma \in \N$ with $k \leq m$ and $\gamma \leq n$, an open set $\mathcal{E}\subseteq\R^n$ such that $\mathcal{D}\subseteq\mathcal{E}$, a transformation $\bs{\Phi}: \mathcal{E} \to \R^n$ that is a diffeomorphism between $\mathcal{E}$ and an open subset of $\R^n$, a controller $\mb{k}_{\mathrm{fbl}}: \mathcal{X} \times \R^k \to \R^m$, a controllable pair $(\mb{A}, \mb{B}) \in \R^{\gamma \times \gamma} \times \R^{\gamma \times k}$, and a function $\mb{q}: \bs{\Phi}(\mathcal{D}) \to \R^{n - \gamma}$ satisfying:
\begin{equation}
\label{eqn:fblindef}
    D\bs{\Phi}(\mb{x})(\mb{f}(\mb{x}) + \mb{g}(\mb{x})\mb{k}_{\mathrm{fbl}}(\mb{x}, \mb{v})) = \begin{bmatrix} \mb{A}\bs{\eta} + \mb{B}\mb{v} \\ \mb{q}(\bs{\xi}) \end{bmatrix},
\end{equation}
for all states $\mb{x} \in \mathcal{X}$ and auxiliary control inputs $\mb{v} \in \R^k$, where $\bs{\eta} \in \R^\gamma$, $\mb{z} \in \R^{n - \gamma}$, and $\bs{\xi} \in \R^n$ satisfy $(\bs{\eta}, \mb{z}) = \bs{\xi} =  \bs{\Phi}(\mb{x})$. Note that if $\gamma = n$, the system is full-state feedback linearizable, and the function $\mb{q}$ does not appear in \eqref{eqn:fblindef}. The corresponding system in \textit{normal form} is governed by:
\begin{align}
    \label{eqn:outputdyn}
    \dot{\bs{\xi}} \triangleq \begin{bmatrix}\dot{\bs{\eta}} \\ \dot{\mb{z}}\end{bmatrix} = \begin{bmatrix} \mb{f}_{\bs{\eta}}(\bs{\xi}) \\ 
    \mb{q}(\bs{\xi}) \end{bmatrix} + \begin{bmatrix} \mb{g}_{\bs{\eta}}(\bs{\xi}) \\ \mb{0}_{n - \gamma} \end{bmatrix} \mb{u} = \mb{f}_{\bs{\xi}}(\bs{\xi}) + \mb{g}_{\bs{\xi}}(\bs{\xi})\mb{u},
\end{align}
for \textit{normal state} signal $\bs{\xi}$, \textit{output} signal $\bs{\eta}$, \textit{zero-coordinate} signal $\mb{z}$, and control input signal $\mb{u}$, with $\mb{f}_{\bs{\eta}}: \bs{\Phi}(\mathcal{D}) \to \R^\gamma$ and $\mb{g}_{\bs{\eta}}: \bs{\Phi}(\mathcal{D}) \to \R^{\gamma \times m}$ defined such that $\mb{f}_{\bs{\xi}}: \bs{\Phi}(\mathcal{D}) \to \R^n$ and $\mb{g}_{\bs{\xi}}: \bs{\Phi}(\mathcal{D}) \to \R^{n \times m}$ satisfy:
\begin{equation}
    D\bs{\Phi}(\bs{\Phi}^{-1}(\bs{\xi}))(\mb{f}(\bs{\Phi}^{-1}(\bs{\xi})) + \mb{g}(\bs{\Phi}^{-1}(\bs{\xi}))\mb{u}) = \mb{f}_{\bs{\xi}}(\bs{\xi}) + \mb{g}_{\bs{\xi}}(\bs{\xi})\mb{u}, \nonumber
\end{equation}
for all $\bs{\xi} \in \bs{\Phi}(\mathcal{D})$ and $\mb{u} \in \R^m$.
\end{definition}

\begin{remark}
As shown in \cite{grizzle1988feedback}, feedback linearizability of a continuous time system does not guarantee feedback linearizability of the resulting sampled-data system, even when using approximate discrete time models. In particular, this property may be lost due to a change of coordinates. The preservation of this property motivates studying the evolution of the normal form system in the sampled-data context. 
\end{remark}

As we will consider the control design process for the normal form system \eqref{eqn:outputdyn}, it is useful to define the set:
\begin{equation}
    \mathcal{Z}_{\bs{\xi}} = \{(\bs{\xi},\mb{u}) \in \bs{\Phi}(\mathcal{X}) \times \R^m  ~|~ (\bs{\Phi}^{-1}(\bs{\xi}),\mb{u})\in\mathcal{Z}\},
\end{equation}
noting that for every state-input pair $(\bs{\xi}_0,\mb{u}_0)\in\mathcal{Z}_{\bs{\xi}}$, there exists a unique solution $\bs{\psi}:[0,T_{\mathrm{max}}]\to\R^n$ satisfying:
\begin{align}
\label{eqn:varpsiode}
    \dot{\bs{\psi}}(t) &= \mb{f}_{\bs{\xi}}(\bs{\psi}(t)) + \mb{g}_{\bs{\xi}}(\bs{\psi}(t))\mb{u}_0 \quad \forall t \in (0, T_{\max}),\\
    \bs{\psi}(0) &= \bs{\xi}_0. \label{eqn:varpsiic}
\end{align}
For $h \in (0, T_{\max}]$, a controller $\mb{k}:\bs{\Phi}(\mathcal{X})\to\R^m$ is an $h$-admissible controller if the corresponding controller $\mb{k}':\mathcal{X}\to\R^m$ given by $\mb{k}'(\mb{x}) = \mb{k}(\bs{\Phi}(\mb{x}))$ for all $\mb{x}\in\mathcal{X}$ is $h$-admissible. A controller $\mb{k}_{\mathrm{aux}}: \bs{\Phi}(\mathcal{X}) \to \R^m$ is an \textit{$h$-admissible auxiliary controller} if $\mb{k}: \bs{\Phi}(\mathcal{X}) \to \R^m$ given by: 
\begin{equation}\label{eqn:total-fbl}
    \mb{k}(\bs{\xi}) = \mb{k}_{\mathrm{fbl}}(\bs{\Phi}^{-1}(\bs{\xi}), \mb{k}_{\mathrm{aux}}(\bs{\xi})),
\end{equation}
is an $h$-admissible controller.


\section{Sampled-Data Control}
\label{sec:sampdata}

This section provides a review of the sampled-data control setting, in which inputs are applied to the system with a zero-order hold. In this setting, the set of possible sample periods is given by $I = (0, T_{\mathrm{max}}]$. Given a sample period $h\in I$ and an $h$-admissible controller $\mb{k}:\bs{\Phi}(\mathcal{X})\to\R^m$, the normal state and control input signals in \eqref{eqn:outputdyn} satisfy:
\begin{equation}
    \mb{u}(t) = \mb{k}(\bs{\xi}(t_k)) \quad \forall t\in[t_k,t_{k+1}),
\end{equation}
with sample times satisfying $t_{k+1}-t_k = h$ for all $k \in \bb{Z}_+$. In this setting, the evolution of the system over a sample period is given by the \textit{exact state discrete map} $\mb{F}_h^{e,\mb{x}}: \mathcal{Z} \to \mathcal{D}$ and \textit{exact normal discrete map} $\mb{F}_h^{e,\bs{\xi}}: \mathcal{Z}_{\bs{\xi}} \to \bs{\Phi}(\mathcal{D})$, defined as:
\begin{align}
    \mb{F}^{e,\mb{x}}_h(\mb{x}_0,\mb{u}_0) &= \mb{x}_0 + \int_{0}^{h} [\mb{f}(\bs{\varphi}(\tau)) + \mb{g}(\bs{\varphi}(\tau))\mb{u}_0]~ \mathrm{d}\tau, \\
    \mb{F}^{e, \bs{\xi}}_h(\bs{\xi}_0,\mb{u}_0) &= \bs{\xi}_0 + \int_{0}^{h} [\mb{f}_{\bs{\xi}}(\bs{\psi}(\tau)) + \mb{g}_{\bs{\xi}}(\bs{\psi}(\tau))\mb{u}_0]~ \mathrm{d}\tau,
\end{align}
for all state-input pairs $(\mb{x}_0, \mb{u}_0) \in \mathcal{Z}$ and all normal state-input pairs $(\bs{\xi}_0, \mb{u}_0) \in \mathcal{Z}_{\bs{\xi}}$. The exact maps are related by:
\begin{equation}
    \mb{F}_h^{e, \bs{\xi}}(\bs{\xi}_0, \mb{u}_0) = \bs{\Phi}(\mb{F}_h^{e, \mb{x}}(\bs{\Phi}^{-1}(\bs{\xi}_0), \mb{u}_0)),
\end{equation}
for all normal state-input pairs $(\bs{\xi}_0, \mb{u}_0) \in \mathcal{Z}_{\bs{\xi}}$.

\begin{remark}
While an equivalence between the exact state discrete map and exact normal discrete map is achieved via the diffeomorphism $\bs{\Phi}$, it is useful to define both maps as the notion of stability we consider for sampled-data systems is defined for a particular exact map.
\end{remark}

We call a family of controllers $\{ \mb{k}_h: \bs{\Phi}(\mathcal{X}) \to \R^m ~|~ h \in I \}$ a \textit{family of admissible controllers} if there is an $h^* \in I$ such that for each $h\in (0, h^*)$, $\mb{k}_h$ is $h$-admissible. This enables the following definition:

\begin{definition}[\textit{Exact Families}]
For a family of admissible controllers $\{ \mb{k}_h: \bs{\Phi}(\mathcal{X}) \to \R^m ~|~ h \in I \}$, we define the \textit{exact state family} $\{ (\mb{k}_h\circ\bs{\Phi}, \mb{F}_h^{e,\mb{x}}) ~|~ h \in I \}$ and \textit{exact normal family} $\{ (\mb{k}_h, \mb{F}_h^{e,\bs{\xi}}) ~|~ h \in I \}$ of controller-map pairs.
\end{definition}
For all $h\in I$ such that $\mb{k}_h$ is $h$-admissible, the recursion $\bs{\xi}_{k + 1} = \mb{F}^{e,\bs{\xi}}_h(\bs{\xi}_k, \mb{k}_h(\bs{\xi}_k))\in\bs{\Phi}(\mathcal{X})$ is well-defined for all $\bs{\xi}_0\in\bs{\Phi}(\mathcal{X})$ and $k\in\mathbb{Z}_+$. In practice, closed-form expressions for these maps are rarely obtainable, suggesting the use of approximations in the control synthesis process. While there are many approaches to approximating this map, we will use the following approximation of the exact normal discrete map:

\begin{definition}[\textit{Euler Approximation Family}]
For every sample period $h \in I$, define the map $\mb{F}_h^{a,\bs{\xi}}: \mathcal{Z}_{\bs{\xi}} \to \R^n$ as:
\begin{equation}
    \mb{F}_h^{a,\bs{\xi}}(\bs{\xi}_0, \mb{u}_0) = \bs{\xi}_0 + h(\mb{f}_{\bs{\xi}}(\bs{\xi}_0) + \mb{g}_{\bs{\xi}}(\bs{\xi}_0)\mb{u}_0),
\end{equation}
for all $(\bs{\xi}_0, \mb{u}_0) \in \mathcal{Z}_{\bs{\xi}}$. For a family of admissible controllers $\{ \mb{k}_h: \bs{\Phi}(\mathcal{X}) \to \R^m ~|~ h \in I \}$, the corresponding \textit{Euler approximation family} of controller-map pairs is:
\begin{equation}
    \{ (\mb{k}_h, \mb{F}_h^{a,\bs{\xi}}) ~|~ h \in I \}.
\end{equation}
\end{definition}
The motivation behind this particular approximation is preserving the strict feedback nature \cite{nevsic2006stabilization} of the normal form. For $h \in I$, we can also define $\mb{F}_h^{a, \bs{\eta}}: \bs{\mathcal{Z}}_{\bs{\xi}} \to \R^\gamma$ and $\mb{F}_h^{a, \mb{z}}: \bs{\Phi}(\mathcal{X}) \to \R^{n - \gamma}$ such that for all $(\bs{\xi}, \mb{u}) = ((\bs{\eta}, \mb{z}), \mb{u}) \in \mathcal{Z}_{\bs{\xi}}$:
\begin{equation}
    \mb{F}_h^{a, \bs{\xi}}(\bs{\xi}, \mb{u}) =  \begin{bmatrix} \mb{F}_h^{a, \bs{\eta}}(\bs{\xi}, \mb{u}) \\ \mb{F}_h^{a, \mb{z}}(\bs{\xi}) \end{bmatrix} = \begin{bmatrix} \bs{\eta} + h(\mb{f}_{\bs{\eta}}(\bs{\xi}) + \mb{g}_{\bs{\eta}}(\bs{\xi})\mb{u}) \\ \mb{z} + h\mb{q}(\bs{\xi}) \end{bmatrix}.\nonumber
\end{equation}

\begin{remark}
Under the current assumptions, there may be an $h\in I$ such that the controller $\mb{k}_h$ is $h$-admissible but the recursion $\bs{\xi}_{k + 1} = \mb{F}^{a,\bs{\xi}}_h(\bs{\xi}_k, \mb{k}_h(\bs{\xi}_k))$ is not well-defined for all $\bs{\xi}_0\in\bs{\Phi}(\mathcal{X})$ and $k\in\mathbb{Z}_+$. This is due to the definition of this map enabling $\bs{\xi}_k\notin \bs{\Phi}(\mathcal{X})$ for some $k>0$. While our results do not need the recursion of the Euler approximation family be well-defined, this can be achieved by extending the domains of $\mb{f}_{\bs{\xi}}$, $\mb{g}_{\bs{\xi}}$, and $\mb{k}_h$ to $\R^n$. 
\end{remark}

Defining class $\mathcal{K}$ ($\mathcal{K}_\infty$) and $\mathcal{KL}$ ($\mathcal{KL}_\infty$) comparison functions as in \cite{freeman2008robust, nevsic1999sufficient}, the following definition characterizes how accurately an approximate map captures the exact map:

\begin{definition}[\textit{One-Step Consistency}]
 A family $\{ (\mb{k}_h,\mb{F}_h): h \in I \}$ is \textit{one-step consistent} with $\{(\mb{k}_h, \mb{F}_h^{e,\bs{\xi}}) ~|~ h \in I \}$ if, for each compact set $K\subseteq\bs{\Phi}(\mathcal{X})$, there exist a function $\rho\in\K_\infty$ and $h^*\in I$ such that for all $\bs{\xi}\in K$ and $h\in(0,h^*)$, we have:
 \begin{equation}\label{eqn:one-step-cons}
     \Vert \mb{F}_h^{e,\bs{\xi}}(\bs{\xi},\mb{k}_h(\bs{\xi}))-\mb{F}_h(\bs{\xi},\mb{k}_h(\bs{\xi})) \Vert \leq h\rho(h).
 \end{equation}
\end{definition}

The following lemma (slightly modified from Lemma 1 in \cite{nevsic1999sufficient}, with proof in the appendix) relates the Euler approximation family and one-step consistency: 

\begin{lemma}
\label{lem:locliponestep}
Suppose $\mb{f}_{\bs{\xi}}$ and $\mb{g}_{\bs{\xi}}$ are locally Lipschitz continuous on $\bs{\Phi}(\mathcal{X})$. Consider a family of admissible controllers $\{ \mb{k}_h: \bs{\Phi}(\mathcal{X}) \to \R^m ~|~ h \in I \}$ and suppose that for any compact set $K\subset\bs{\Phi}(\mathcal{X})$ there exist $h^* \in I$ and a bound $M \in \R_{++}$ such that for every sample time $h \in (0,h^*)$, the controller $\mb{k}_h$ is bounded by $M$ on $K$. Then the family $\{ (\mb{k}_h, \mb{F}_h^{a,\bs{\xi}}) ~|~ h \in I \}$ is one-step consistent with the family $\{ (\mb{k}_h, \mb{F}_h^{e,\bs{\xi}}) ~|~ h \in I \}$.
\end{lemma}

\noindent We note that if $\mb{f}$, $\mb{g}$, and $D\bs{\Phi}$ are locally Lipschitz continuous on $\mathcal{X}$, the first condition of Lemma \ref{lem:locliponestep} is met. We consider the following stability property, defined for both the exact state discrete map and the exact normal discrete map:

\begin{definition}[\textit{Practical Stability}]\label{def:prac-stab}
Let $\beta\in\cal{K}\cal{L}_{\infty}$ and $N \subseteq \R^n$ be an open set containing the origin. A family $\{ (\mb{k}_h,\mb{F}_h): h \in I \}$ is $(\beta,N)$-practically stable if for each $R\in\R_{++}$, there exists an $h^*\in I$ such that for each sample period $h \in (0, h^*)$, initial state $\bs{\zeta}_0 \in N$, and number of steps $k \in \Z_+$, the recursion $\bs{\zeta}_{k + 1} = \mb{F}_h(\bs{\zeta}_k, \mb{k}_h(\bs{\zeta}_k))$ is well-defined and:
\begin{equation}
    \Vert \bs{\zeta}_k \Vert \leq \beta(\Vert \bs{\zeta}_0 \Vert, kh)+R.
\end{equation}
\end{definition}

The following lemma relates the practical stability of the exact normal family and the exact state family. Importantly, it justifies considering the sampled normal form dynamics, which can be feedback linearized, rather than the sampled state dynamics that may not be feedback linearizable.

\begin{lemma}
\label{lem:diffeopracstability}
Suppose that $\mb{0}_n\in\mathcal{X}$ and $\bs{\Phi}(\mb{0}_n)=\mb{0}_n$, and that for any compact sets $K,K'\subset\R^n$, $\bs{\Phi}$ and $\bs{\Phi}^{-1}$ are globally Lipschitz continuous on $K'\cap\mathcal{X}$ and $K\cap\bs{\Phi}(\mathcal{X})$, respectively. If the exact normal family $\{ (\mb{k}_h, \mb{F}_h^{e,\bs{\xi}}) ~|~ h \in I \}$ is $(\beta,N)$-practically stable, then there exist $\beta'\in\mathcal{KL}_\infty$ and a bounded open set $N'\subseteq\R^n$ with $\mb{0}_n\in N'$ such that the exact state family $\{ (\mb{k}_h\circ\bs{\Phi}, \mb{F}_h^{e,\mb{x}}) ~|~ h \in I \}$ is $(\beta',N')$-practically stable.
\end{lemma}

A proof is provided in the appendix. The following class of Lyapunov functions is useful in certifying practical stability:

\begin{definition}[\textit{Asymptotic Stability by Equi-Lipschitz Lyapunov Functions}]\label{def:as-by-equi-lip-lyap}
Consider a family of admissible controllers $\{ \mb{k}_h: \bs{\Phi}(\mathcal{X}) \to \R^m ~|~ h \in I \}$. A family $\{ (\mb{k}_h, \mb{F}_h) ~|~ h \in I \}$ is \textit{asymptotically stable by equi-Lipschitz Lyapunov functions} if for some open set $N \subseteq \bs{\Phi}(\mathcal{X})$ containing the origin and any compact set $K \subseteq N$, there exist $h^* \in I$, comparison functions $\alpha_1, \alpha_2 \in \mathcal{K}_\infty$ and $\alpha_3 \in \mathcal{K}$, a family $\{ V_h: \R^n \to \R_+ ~|~ h \in (0, h^*) \}$, and a Lipschitz constant $M \in \R_{++}$ such that:
\begin{align}
\label{eqn:xibounds}
    &\alpha_1(\| \bs{\xi}_1 \|) \leq V_h(\bs{\xi}_1) \leq \alpha_2(\| \bs{\xi}_1 \|),\\ \label{eqn:decrease-bound}
    &V_h(\mb{F}_h(\bs{\xi}_2, \mb{k}_h(\bs{\xi}_2))) - V_h(\bs{\xi}_2) \leq -h\alpha_3(\| \bs{\xi}_2 \|),\\ \label{eqn:lipschitzV}
    &\vert V_h(\bs{\xi}_3) - V_h(\bs{\xi}_4) \vert \leq M \Vert \bs{\xi}_3 - \bs{\xi}_4 \Vert,
\end{align}
for all $\bs{\xi}_1 \in \R^n$, normal states $\bs{\xi}_2 \in N$ and $\bs{\xi}_3, \bs{\xi}_4 \in K$, and sample times $h \in (0, h^*)$.
\end{definition}

These functions serve to connect one-step consistency and practical stability, given by the following result (a local variant of Theorem 2 in \cite{nevsic1999sufficient}, with proof in the appendix): 


\begin{lemma}
\label{lem:exactpracstab}
Consider a family of admissible controllers $\{ \mb{k}_h: \bs{\Phi}(\mathcal{X}) \to \R^m ~|~ h \in I \}$. If the corresponding Euler approximation family $\{ (\mb{k}_h, \mb{F}_h^{a,\bs{\xi}}) ~|~ h \in I \}$ is asymptotically stable by equi-Lipschitz Lyapunov functions, then there exist $\beta \in \mathcal{KL}_\infty$ and a bounded open set $U \subseteq \bs{\Phi}(\mathcal{X})$ with $\mb{0}_n \in U$ such that the family $\{ (\mb{k}_h, \mb{F}_h^{e,\bs{\xi}}) ~|~ h \in I \}$ is $(\beta, N)$-practically stable for any open set $N \subseteq U$ with $\mb{0}_n\in N$.
\end{lemma}

\section{Stabilization}
\label{sec:results}
In this section we present our main results establishing feedback linearization as a method for practically stabilizing sampled-data nonlinear systems. The first result builds on \cite{tabuada2020data} to make a claim on the stabilizability of the output dynamics:

\begin{lemma}
\label{lem:outputstab}
Given a feedback linearizable system satisfying \eqref{eqn:fblindef}, consider $\mb{K} \in \R^{k \times \gamma}$ such that $\mb{A}_{\mathrm{cl}} \triangleq \mb{A} - \mb{B}\mb{K}$ is Hurwitz. Let $\mb{P}_{\bs{\eta}}\in\bb{S}^\gamma_{++}$ solve the continuous time Lyapunov Equation:
\begin{equation}
    \mb{A}_{\mathrm{cl}}^\top\mb{P}_{\bs{\eta}} + \mb{P}_{\bs{\eta}}\mb{A}_{\mathrm{cl}} = -\mb{Q}_{\bs{\eta}},
\end{equation}
for some $\mb{Q}_{\bs{\eta}}\in\bb{S}^\gamma_{++}$. Define the function $V_{\bs{\eta}}:\R^\gamma\to\R_{+}$ as $V_{\bs{\eta}}(\bs{\eta}) = \bs{\eta}^\top\mb{P}_{\bs{\eta}}\bs{\eta}$ for all $\bs{\eta} \in \R^\gamma$. For any $c \in (0, 1)$, there exists $h_{\bs{\eta}}^* \in I$ such that for any $\bs{\eta}_0 \in \R^\gamma$, $\bs{\xi} = (\bs{\eta}, \mb{z}) \in \bs{\Phi}(\mathcal{X})$, and $h \in (0, h_{\bs{\eta}}^*)$, there exists an input $\mb{u}\in\R^m$ such that:
\begin{align}
\label{eqn:etabounds}
    & \lambda_{\min}(\mb{P}_{\bs{\eta}})\Vert\bs{\eta}_0\Vert^2_2 \leq V_{\bs{\eta}}(\bs{\eta}_0)  \leq  \lambda_{\max}(\mb{P}_{\bs{\eta}})\Vert\bs{\eta}_0\Vert^2_2, \\\label{eqn:decrease-upper-bound} & V_{\bs{\eta}}(\mb{F}_h^{a, \bs{\eta}}(\bs{\xi}, \mb{u}))-V_{\bs{\eta}}(\bs{\eta}) \leq -hc\lambda_{\min}(\mb{Q}_{\bs{\eta}})\Vert\bs{\eta}\Vert_2^2. 
\end{align}
\end{lemma}

\begin{proof}
The bounds in \eqref{eqn:etabounds} follow from the definition of $V_{\bs{\eta}}$. Define the auxiliary controller $\mb{k}_{\mathrm{aux}}((\bs{\eta}, \mb{z})) = - \mb{K}\bs{\eta}$ for all $(\bs{\eta},\mb{z})\in\bs{\Phi}(\mathcal{X})$. For the controller $\mb{k}$ defined in \eqref{eqn:total-fbl}, we have:
\begin{align}
    & V_{\bs{\eta}}(\mb{F}_h^{a, \bs{\eta}}(\bs{\xi}, \mb{k}(\bs{\xi}))) - V_{\bs{\eta}}(\bs{\eta}) = V_{\bs{\eta}}((\mb{I}_\gamma + h\mb{A}_{\mathrm{cl}})\bs{\eta}) - V_{\bs{\eta}}(\bs{\eta}) \nonumber\\ &~= -h\bs{\eta}^\top(\mb{A}_{\mathrm{cl}}^\top\mb{P}_{\bs{\eta}}+\mb{P}_{\bs{\eta}}\mb{A}_{\mathrm{cl}} + h  \mb{A}_{\mathrm{cl}}^\top \mb{P}_{\bs{\eta}}\mb{A}_{\mathrm{cl}})\bs{\eta} \nonumber \\ & ~\leq -h(\lambda_{\min}(\mb{Q}_{\bs{\eta}}) - h \lambda_{\max}(\mb{A}_{\mathrm{cl}}^\top \mb{P}_{\bs{\eta}}\mb{A}_{\mathrm{cl}})) \Vert \bs{\eta} \Vert_2^2, \nonumber
\end{align}
for all $\bs{\xi} = (\bs{\eta}, \mb{z}) \in \bs{\Phi}(\mathcal{X})$ and $h \in I$. Picking $h^*_{\bs{\eta}} \in I$ with:
\begin{equation}
    h^*_{\bs{\eta}} \leq (1-c)\frac{\lambda_{\min}(\mb{Q}_{\bs{\eta}})}{\lambda_{\max}(\mb{A}_{\mathrm{cl}}^\top \mb{P}_{\bs{\eta}}\mb{A}_{\mathrm{cl}})},
\end{equation}
implies that for all $h\in(0,h^*_{\bs{\eta}}]$ and $\bs{\xi} \in \bs{\Phi}(\mathcal{X})$, the input $\mb{k}(\bs{\xi})$ satisfies \eqref{eqn:decrease-upper-bound}. 
\end{proof}
We call the function $V_{\bs{\eta}}$ a discrete time Control Lyapunov Function (CLF) for any Euler approximate model of the output dynamics with $h\in(0,h^*_{\bs{\eta}}]$. For each $h \in (0, h_{\bs{\eta}}^*)$, define the set-valued function $\mathcal{U}_h: \bs{\Phi}(\mathcal{X}) \to \mathcal{P}(\R^m)$ as:
\begin{align}
    \mathcal{U}_h(\bs{\xi}) = \left\{ \mb{u}\in\R^m ~|~ (\bs{\xi}, \mb{u}) \in \mathcal{Z}_{\bs{\xi}}; ~\mb{u} ~\mathrm{satisfies}~\eqref{eqn:decrease-upper-bound} ~\mathrm{for}~ h, \bs{\xi} \right\},\nonumber
\end{align}
for all $\bs{\xi} \in \bs{\Phi}(\mathcal{X})$. The next result connects these functions to continuous time stability of the zero-dynamics, implying the conditions of Lemma \ref{lem:exactpracstab} are met for a wider class of controllers than feedback linearizing controllers:


\begin{theorem}
\label{thm:stability}
Let $V_{\bs{\eta}}$ and $h^*_{\bs{\eta}}$ be defined as in Lemma \ref{lem:outputstab}, and assume that $\mb{q}$ is continuously differentiable and the zero-dynamics system governed by the differential equation:
\begin{equation}
\label{eqn:zerodynamics}
    \dot{\mb{z}} = \mb{q}(\mb{0}_\gamma,\mb{z}),
\end{equation}
for zero-coordinate signal $\mb{z}$ is locally exponentially stable to the origin.
Let $\{\mb{k}_h:\bs{\Phi}(\mathcal{X})\to\R^m ~|~ h\in I \}$ be a family of admissible controllers satisfying $\mb{k}_h(\bs{\xi})\in\mathcal{U}_h(\bs{\xi})$ for all $h\in (0,h^*_{\bs{\eta}}]$ and $\bs{\xi}\in\bs{\Phi}(\mathcal{X})$. Then the family $\{(\mb{k}_h,\mb{F}_h^{a,\bs{\xi}})~|~h\in I\}$ is asymptotically stable by equi-Lipschitz Lyapunov functions.
\end{theorem}

\begin{proof}
The local exponential stability of \eqref{eqn:zerodynamics} implies that for any $\mb{Q}_{\mb{z}}\in\mathbb{S}^n_{++}$ and $d \in (0, 1)$, there exist an open neighborhood of the origin $N \subseteq \R^{n - \gamma}$, an $h_{\mb{z}}^* \in I$, a $\mb{P}_{\mb{z}}\in\mathbb{S}^n_{++}$, and a quadratic Lyapunov function $V_{\mb{z}}: \R^{n - \gamma} \to \R_+$ defined as $V_{\mb{z}}(\mb{z}) = \mb{z}^\top\mb{P}_{\mb{z}}\mb{z}$ for all $\mb{z} \in \R^{n - \gamma}$ and satisfying:
\begin{align}
\label{eqn:zbounds}
    & \lambda_{\min}(\mb{P}_{\mb{z}})\Vert\mb{z}_0\Vert^2_2 \leq V_{\mb{z}}(\mb{z}_0)  \leq  \lambda_{\max}(\mb{P}_{\mb{z}})\Vert\mb{z}_0\Vert^2_2, \\ & V_{\mb{z}}(\mb{F}_h^{a,\mb{z}}((\mb{0}_\gamma,\mb{z})))-V_{\mb{z}}(\mb{z}) \leq -hd\lambda_{\min}(\mb{Q}_{\mb{z}})\Vert\mb{z}\Vert_2^2,
\end{align}
for all $\mb{z}_0 \in \R^{n - \gamma}$, $\mb{z} \in N$, and $h \in (0, h_{\mb{z}}^*)$. Construction of $V_{\mb{z}}$ follows the steps of Lemma \ref{lem:outputstab} with the linearization of $\mb{q}$ at the origin. Let $\sigma \in \R_{++}$ be a coefficient to be specified later. Define the composite Lyapunov function $V:\R^n\to\R_+$ as:
\begin{equation}
    V(\bs{\xi}) = \sigma V_{\bs{\eta}}(\bs{\eta}) + V_{\mb{z}}(\mb{z}),
\end{equation}
for all $\bs{\xi} = (\bs{\eta}, \mb{z}) \in \R^n$. First, note that:
\begin{align}
    \min{\{ \sigma\lambda_{\min}(\mb{P}_{\bs{\eta}}), \lambda_{\min}(\mb{P}_\mb{z}) \}} \Vert \bs{\xi} \Vert_2^2  \leq V(\bs{\xi})\nonumber\\
    \leq \underbrace{\max{\{ \sigma\lambda_{\max}(\mb{P}_{\bs{\eta}}), \lambda_{\max}(\mb{P}_\mb{z}) \}}}_{\triangleq \mu} \Vert \bs{\xi} \Vert_2^2 ,
\end{align}
for all $\bs{\xi} \in \R^n$. Second, note that:
\begin{align}
    \Vert \nabla V(\bs{\xi}) \Vert_2 &\leq 2\left( \sigma \lambda_{\max}(\mb{P}_{\bs{\eta}}) \Vert \bs{\eta} \Vert_2 + \lambda_{\max}(\mb{P}_{\mb{z}}) \Vert \mb{z} \Vert_2 \right) \nonumber \\
    &\leq 2(\mu \Vert \bs{\xi} \Vert_2 + \mu \Vert \bs{\xi} \Vert_2) = 4 \mu \Vert \bs{\xi} \Vert_2, \nonumber
\end{align}
for all $\bs{\xi} = (\bs{\eta}, \mb{z}) \in \bs{\Phi}(\mathcal{X})$, implying that for any compact set $K\subset \bs{\Phi}(\mathcal{X})$, we have:
\begin{equation}
    \vert V(\bs{\xi}_1)-V(\bs{\xi}_2)\vert \leq 4 \mu \left(\max_{\bs{\xi}\in K} \Vert\bs{\xi}\Vert_2\right) \Vert \bs{\xi}_1 - \bs{\xi}_2 \Vert_2,
\end{equation}
for all $\bs{\xi}_1, \bs{\xi}_2 \in K$. Third, define a bounded open set $N_{\bs{\xi}} \subset \R^n$ with closure $\mathrm{cl}(N_{\bs{\xi}}) \subset \bs{\Phi}(\mathcal{X}) \cap (\R^\gamma \times N)$, let $L_{\mb{q}} \in \R_{++}$ be a global Lipschitz constant of $\mb{q}$ on $N_{\bs{\xi}}$, and let $h_1^* = \min{\{ h_{\bs{\eta}}^*, h_{\mb{z}}^* \}}$. For all $\bs{\xi} = (\bs{\eta}, \mb{z}) \in N_{\bs{\xi}}$ and $h \in (0, h_1^*)$, note:
\begin{align}
    & V(\mb{F}^{a,\bs{\xi}}_h(\bs{\xi},\mb{k}_h(\bs{\xi}))) - V(\bs{\xi})\nonumber\\
    &~= \sigma(V_{\bs{\eta}}(\mb{F}^{a,\bs{\eta}}_h(\bs{\xi},\mb{k}_h(\bs{\xi}))) - V_{\bs{\eta}}(\bs{\eta})) + V_{\mb{z}}(\mb{F}^{a,\mb{z}}_h(\bs{\xi})) - V_{\mb{z}}(\mb{z}) \nonumber\\
    &~ \leq -\sigma hc\lambda_{\min}(\mb{Q}_{\bs{\eta}})\Vert\bs{\eta}\Vert_2^2 + V_{\mb{z}}(\mb{F}^{a,\mb{z}}_h((\mb{0}_\gamma, \mb{z}))) - V_{\mb{z}}(\mb{z}) \nonumber \\
    &\quad + V_{\mb{z}}(\mb{F}^{a,\mb{z}}_h((\bs{\eta}, \mb{z}))) - V_{\mb{z}}(\mb{F}^{a,\mb{z}}_h((\mb{0}_\gamma, \mb{z}))) \nonumber \\
    &~\leq -\sigma hc\lambda_{\min}(\mb{Q}_{\bs{\eta}})\Vert\bs{\eta}\Vert_2^2 - hd\lambda_{\min}(\mb{Q}_{\mb{z}})\Vert\mb{z}\Vert_2^2 \nonumber\\ 
    &\quad + 2h\mb{z}^\top\mb{P}_{\mb{z}}(\mb{q}((\bs{\eta},\mb{z}))-\mb{q}((\mb{0}_{\gamma},\mb{z}))) \nonumber \\
    &\quad + h^2(\mb{q}((\bs{\eta},\mb{z}))^\top\mb{P}_{\mb{z}}\mb{q}((\bs{\eta},\mb{z})) - \mb{q}((\mb{0}_\gamma,\mb{z}))^\top\mb{P}_{\mb{z}}\mb{q}((\mb{0}_\gamma,\mb{z}))) \nonumber \\
    &~ \leq  -\sigma hc\lambda_{\min}(\mb{Q}_{\bs{\eta}})\Vert\bs{\eta}\Vert_2^2 - hd\lambda_{\min}(\mb{Q}_{\mb{z}})\Vert\mb{z}\Vert_2^2 \nonumber \\
    &\quad + 2h\lambda_{\max}(\mb{P}_{\mb{z}}) L_{\mb{q}} \Vert\bs{\eta}\Vert_2\Vert \mb{z} \Vert_2 + h^2 \lambda_{\max}(\mb{P}_{\mb{z}}) L_{\mb{q}} \Vert \bs{\xi} \Vert_2^2\nonumber\\
    &~= -h\begin{bmatrix} \Vert \bs{\eta} \Vert_2 \\ \Vert \mb{z} \Vert_2 \end{bmatrix}^\top \underbrace{\begin{bmatrix} \omega_{\bs{\eta}}(\sigma,h) & -\omega_{\times} \\ -\omega_{\times} & \omega_{\mb{z}}(h) \end{bmatrix}}_{\triangleq \bs{\Omega}_\sigma(h)} \begin{bmatrix} \Vert \bs{\eta} \Vert_2 \\ \Vert \mb{z} \Vert_2 \end{bmatrix},
\end{align}
where $\omega_{\bs{\eta}}(\sigma,h) = \sigma c\lambda_{\min}(\mb{Q}_{\bs{\eta}}) - h\lambda_{\max}(\mb{P}_{\mb{z}}) L_{\mb{q}} $, $\omega_{\times} = \lambda_{\max}(\mb{P}_{\mb{z}}) L_{\mb{q}}$, and $\omega_{\mb{z}}(h) = d\lambda_{\min}(\mb{Q}_{\mb{z}}) - h\lambda_{\max}(\mb{P}_{\mb{z}}) L_{\mb{q}}$. Pick $h_2^* \in (0, h_1^*]$ such that $h_2^* < d\lambda_{\min}(\mb{Q}_{\mb{z}}) / \omega_{\times}$ and fix $\sigma$ with:
\begin{equation}
    \sigma > \frac{\omega_\times^2 / \omega_{\mb{z}}(h_2^*) + h_2^* \lambda_{\max}(\mb{P}_{\mb{z}}) L_{\mb{q}}}{c\lambda_{\min}(\mb{Q}_{\bs{\eta}})}, \nonumber
\end{equation}
to ensure that $\bs{\Omega}_{\sigma}(h) \in \bb{S}^n_{++}$ for all $h \in [0, h_2^*]$. The composition $\lambda_{\min} \circ \bs{\Omega}_\sigma$ is continuous and $\R_{++}$-valued for all $h \in [0, h_2^*]$ as $\bs{\Omega}_\sigma$ is an affine function. Therefore:
\begin{align}
    &V(\mb{F}^{a,\bs{\xi}}_h(\bs{\xi},\mb{k}_h(\bs{\xi}))) - V(\bs{\xi}) \leq -h \lambda_{\min}(\bs{\Omega}_\sigma(h)) \Vert \bs{\xi}\Vert_2^2\nonumber\\
    &\quad\leq -h \left( \min_{h' \in [0, h_2^*] } \lambda_{\min}(\bs{\Omega}_\sigma(h')) \right)  \Vert \bs{\xi} \Vert_2^2,
\end{align}
for all $\bs{\xi} \in N_{\bs{\xi}}$ and $h \in (0, h_2^*]$.
\end{proof}
This result implies the practical stability of an exact family built using feedback linearizing controllers. Specifically, with the controller $\mb{k}$ used in Lemma \ref{lem:outputstab}, if there exists an $h_0 \in I$ such that $\mb{k}$ is $h_0$-admissible, then the exact family $\{ (\mb{k}, \mb{F}_h^{e, \bs{\xi}}) ~|~ h \in I \}$ is $(\beta, N)$-practically stable for some $\beta \in \mathcal{KL}_\infty$ and open set $N \subseteq \bs{\Phi}(\mathcal{X})$ containing the origin.

\section{Optimization-Based Sampled-Data Control}
\label{sec:sim} 
As motivated in \cite[p. 6]{freeman2008robust}, the performance of a feedback linearizing controller can be improved upon by optimizing control inputs subject to stability constraints imposed via the CLF $V_{\bs{\eta}}$ found in Lemma \ref{lem:outputstab}. We note that the existence of the feedback linearizing controller ensures the function $V_{\bs{\eta}}$ is also a CLF for the continuous time output dynamics in \eqref{eqn:outputdyn}. For sample period $h \in I$, continuous time design yields a controller $\mb{k}_h^{\mathrm{qp}}: \bs{\Phi}(\mathcal{X}) \to \R^m$ specified by the following quadratic program (QP):
\begin{align}
\label{eqn:CLF-QP}
\tag{CLF-QP}
  \mb{k}_h^{\textrm{qp}}(\bs{\xi}) =   \argmin_{\mb{u} \in \R^m} &~ \| \mb{u} \|_2^2\nonumber\\
    \mathrm{s.t.}~ \nabla V_{\bs{\eta}}&(\bs{\eta})^\top (\mb{f}_{\bs{\eta}}(\bs{\xi})+\mb{g}_{\bs{\eta}}(\bs{\xi})\mb{u}) \leq -\lambda_{\min}(\mb{Q}_{\bs{\eta}})\Vert \bs{\eta} \Vert_2^2 \nonumber,
\end{align}
for all $\bs{\xi} = (\bs{\eta}, \mb{z}) \in \bs{\Phi}(\mathcal{X})$. This controller often displays degradation in performance with sample frequency limitations, motivating the specification of a sampled-data controller. For $h\in(0,h^*_{\bs{\eta}}]$, using the Euler approximate model $\mb{F}_h^{a,\bs{\eta}}$, consider a controller $\mb{k}_h^{\mathrm{qcqp}}: \bs{\Phi}(\mathcal{X}) \to \R^m$ specified by the following quadratically constrained quadratic program (QCQP):
\begin{align}
\label{eqn:CLF-QCQP}
\tag{CLF-QCQP}
    &\mb{k}_h^{\textrm{qcqp}}(\bs{\xi}) = \argmin_{\mb{u} \in \R^m} \| \mb{u} \|_2^2\nonumber\\
    &\quad\mathrm{s.t.}~ V_{\bs{\eta}}(\mb{F}_h^{a,\bs{\eta}}(\bs{\xi}, \mb{u})) - V_{\bs{\eta}}(\bs{\eta}) \leq -hc\lambda_{\min}(\mb{Q}_{\bs{\eta}}) \Vert \bs{\eta} \Vert_2^2 \nonumber \\
    &\qquad\quad~ = \argmin_{\mb{u} \in \R^m} \| \mb{u} \|_2^2\nonumber\\
    &\quad\mathrm{s.t.}~ \mb{u}^\top \bs{\Lambda}_h(\bs{\xi}) \mb{u} + 2\bs{\lambda}_h(\bs{\xi})^\top \mb{u} + l_h(\bs{\xi}) \leq 0, \nonumber
\end{align}
for all $\bs{\xi} = (\bs{\eta}, \mb{z}) \in \bs{\Phi}(\mathcal{X})$ where $\bs{\Lambda}_h: \bs{\Phi}(\mathcal{X}) \to \bb{S}^m_{+}$, $\bs{\lambda}_h: \bs{\Phi}(\mathcal{X}) \to \R^m$, and $l_h: \bs{\Phi}(\mathcal{X}) \to \R$ are defined with $\mb{P}_{\bs{\eta}}$, $\mb{Q}_{\bs{\eta}}$, and $c$ from Lemma \ref{lem:outputstab} as:
\begin{align}
    \bs{\Lambda}_h(\bs{\xi}) &= h\mb{g}_{\bs{\eta}}(\bs{\xi})^\top\mb{P}_{\bs{\eta}}\mb{g}_{\bs{\eta}}(\bs{\xi}),\\
    \bs{\lambda}_h(\bs{\xi}) &= \mb{g}_{\bs{\eta}}(\bs{\xi})^\top \mb{P}_{\bs{\eta}}(\bs{\eta} + h\mb{f}_{\bs{\eta}}(\bs{\xi})),\\
    l_h(\bs{\xi}) &= \mb{f}_{\bs{\eta}}(\bs{\xi})^\top \mb{P}_{\bs{\eta}}(2\bs{\eta} + h\mb{f}_{\bs{\eta}}(\bs{\xi})) + c\lambda_{\min}(\mb{Q}_{\bs{\eta}}) \|\bs{\eta}\|_2^2,
\end{align}
for all normal states $\bs{\xi} = (\bs{\eta}, \mb{z}) \in \bs{\Phi}(\mathcal{X})$. Note that for any normal state $\bs{\xi} \in \bs{\Phi}(\mathcal{X})$, the input $\mb{k}(\bs{\xi})$ is in the feasible set of the corresponding optimization problem, and as the feasible set is closed and the $\Vert \mb{k}(\bs{\xi}) \Vert_2^2$-sublevel set of the continuous objective function is compact, there exists a minimizer in this sublevel set. Moreover, since the objective function is strictly convex and the feasible set is convex, this minimizer is unique.

For each $h \in (h_{\bs{\eta}}^*, T_{\max}]$, define $\mb{k}_h^{\mathrm{qcqp}}: \bs{\Phi}(\mathcal{X}) \to \R^m$ arbitrarily. If $\{ \mb{k}_h^{\mathrm{qcqp}} ~|~ h \in I \}$ is a family of admissible controllers, then the exact family $\{ (\mb{k}_h^{\mathrm{qcqp}}, \mb{F}_h^{e, \bs{\xi}}) ~|~ h \in I \}$ is $(\beta, N)$-practically stable for some $\beta \in \mathcal{KL}_\infty$ and open set $N \subseteq \bs{\Phi}(\mathcal{X})$ by Theorem \ref{thm:stability}. This follows as the feasibility of the feedback linearizing control input implies the family $\{ (\mb{k}_h^{\mathrm{qcqp}}, \mb{F}_h^{a, \bs{\xi}}) ~|~ h \in I \}$ is asymptotically stable by the same Lyapunov functions as the family $\{ (\mb{k}, \mb{F}_h^{a, \bs{\xi}}) ~|~ h \in I \}$. 

To illustrate the advantage of sampled-data design, consider the following system with exponentially stable zero-dynamics:
\begin{equation}
\label{eqn:simsys}
    \dot{\eta}_1 = \eta_2, \quad \dot{\eta}_2 = 10\sin(\eta_1) + u, \quad \dot{z} = \eta_1^2 -z, 
\end{equation}
where $(\eta_1, \eta_2)$ denote the output signal, $z$ denotes the zero-coordinate signal, and $u$ denotes the control signal. For $\mb{K}=\begin{bmatrix}1/2 & \sqrt{3}/2\end{bmatrix}$, $\mb{Q}_{\bs{\eta}} = \mb{I}_2$, $c=0.5$, $h = 0.2$, and initial condition $(1, 0, 1)$, the \eqref{eqn:CLF-QP} fails to stabilize the system, while the \eqref{eqn:CLF-QCQP} achieves stability (see Figure \ref{fig:sim}).

\begin{figure}
    \centering
    \includegraphics[width=0.475\textwidth]{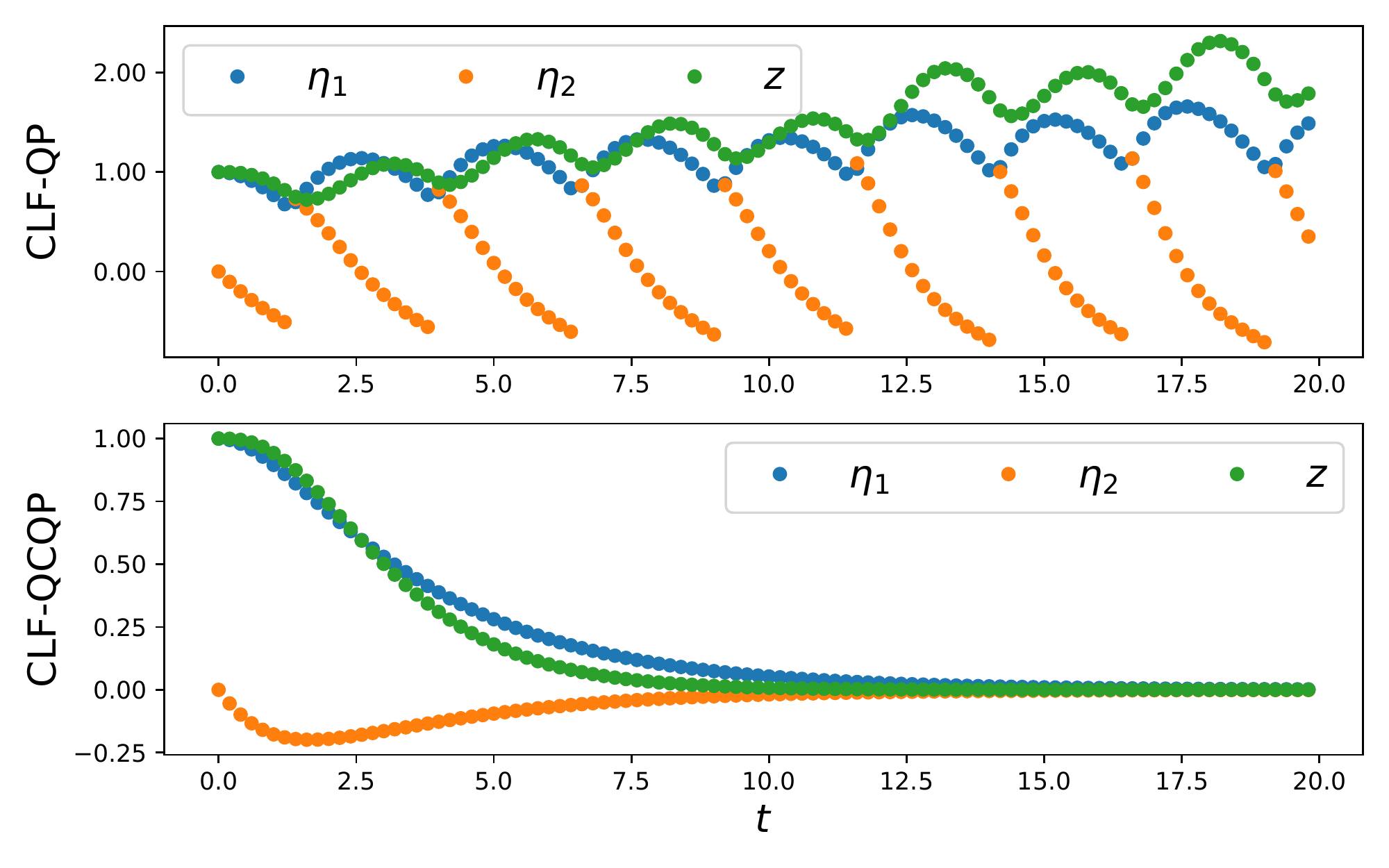}
    \caption{When inputs are applied with a zero-order hold, the controller \eqref{eqn:CLF-QP} designed with continuous time models does not stabilize the system \eqref{eqn:simsys} (Top), while the controller \eqref{eqn:CLF-QCQP} does (Bottom). Simulation code listed at \href{https://bit.ly/CLF-QCQP}{https://bit.ly/CLF-QCQP}.} 
    \label{fig:sim}
\end{figure}

\section{Conclusions}
\label{sec:conclusions}
We presented an approach for sampled-data control synthesis utilizing feedback linearization and CLFs.  In particular, we used the feedback linearizability of a system's continuous time model to yield a discrete time CLF for the Euler approximate discrete time model. We specify a controller with this CLF via a convex optimization problem that improves performance over a continuous time counterpart. Future work will extend this work to notions of safety and Control Barrier Functions.

\section*{Proofs of Lemmas}
\begin{proof}[Proof of Lemma \ref{lem:locliponestep}]
Consider a compact set $K \subset \bs{\Phi}(\mathcal{X})$ and corresponding $h^*\in I$ and $M \in \R_{++}$, and fix a sample period $h \in (0,h^*)$. By assumption, $\mb{k}_h$ is bounded on $K$, and since $\mb{f}_{\bs{\xi}}$ and $\mb{g}_{\bs{\xi}}$ are continuous, $\mb{f}_{\bs{\xi}}$ and $\mb{g}_{\bs{\xi}}$ are also bounded on $K$, implying there exists a bound $M' \in \R_{++}$ with:
\begin{equation}
    \| \mb{f}_{\bs{\xi}}(\bs{\xi}_2) + \mb{g}_{\bs{\xi}}(\bs{\xi}_2)\mb{k}_h(\bs{\xi}_1) \| \leq M', \nonumber
\end{equation}
for all normal states $\bs{\xi}_1, \bs{\xi}_2 \in K$. As $\mb{f}_{\bs{\xi}}$ and $\mb{g}_{\bs{\xi}}$ are locally Lipschitz continuous over the compact set $K$, it follows that $\mb{f}_{\bs{\xi}}$ and $\mb{g}_{\bs{\xi}}$ are globally Lipschitz continuous over $K$. Therefore:
\begin{align}
    & \| \mb{f}_{\bs{\xi}}(\bs{\xi}_2) + \mb{g}_{\bs{\xi}}(\bs{\xi}_2)\mb{k}_h(\bs{\xi}_1) - ( \mb{f}_{\bs{\xi}}(\bs{\xi}_1) + \mb{g}(\bs{\xi}_1)\mb{k}_h(\bs{\xi}_1)) \|\nonumber \\
    &~ \leq \| \mb{f}_{\bs{\xi}}(\bs{\xi}_2) - \mb{f}_{\bs{\xi}}(\bs{\xi}_1) \| + \| \mb{g}_{\bs{\xi}}(\bs{\xi}_2) - \mb{g}_{\bs{\xi}}(\bs{\xi}_1) \| \| \mb{k}_h(\bs{\xi}_1) \| \nonumber\\
    &~ \leq (L_{\mb{f}_{\bs{\xi}}} + L_{\mb{g}_{\bs{\xi}}} M) \| \bs{\xi}_2 - \bs{\xi}_1 \| = \rho(\| \bs{\xi}_2 - \bs{\xi}_1 \|), \nonumber
\end{align}
for all states $\bs{\xi}_1, \bs{\xi}_2 \in K$, where $L_{\mb{f}_{\bs{\xi}}}, L_{\mb{g}_{\bs{\xi}}} \in \R_{++}$ are Lipschitz constants for $\mb{f}_{\bs{\xi}}$ and $\mb{g}_{\bs{\xi}}$, respectively, and $\rho \in \mathcal{K}_\infty$ satisfies $\rho(r) = (L_{\mb{f}_{\bs{\xi}}} + L_{\mb{g}_{\bs{\xi}}} M)r$ for all $r \in \R_+$. The proof proceeds as the proof of Lemma 1 in \cite{nevsic1999sufficient} by substituting $\mathscr{X} = N(K,\epsilon)\subset\bs{\Phi}(\mathcal{X})$, with proper containment implied for some $\epsilon \in \R_{++}$ as $\bs{\Phi}(\mathcal{X})$ is open, and substituting $T_1^* = \min\{h^*, \epsilon / M' \}$. 
\end{proof}

\begin{proof}[Proof of Lemma \ref{lem:diffeopracstability}]
Let $N'\subseteq \bs{\Phi}^{-1}(N)$ be a bounded open set satisfying $\textrm{cl}(N')\subseteq\mathcal{X}$. As $\textrm{cl}(N')$ is compact and $\bs{\Phi}$ is a homeomorphism between $\mathcal{X}$ and $\bs{\Phi}(\mathcal{X})$, we have $\bs{\Phi}(\textrm{cl}(N')) \subseteq \bs{\Phi}(\mathcal{X})$ is compact and satisfies $\bs{\Phi}(\textrm{cl}(N')) = \textrm{cl}(\bs{\Phi}(N'))$. Fix $\overline{R} \in \R_{++}$ and define the corresponding radii $r, r' \in \R_{++}$ as:
\begin{align}
    r &= \max_{\bs{\xi}_0\in\textrm{cl}(\bs{\Phi}(N'))} \beta(\Vert\bs{\xi}_0\Vert,0)+\overline{R}, & r' &= \max_{\mb{x}_0\in\textrm{cl}(N')} \Vert \mb{x}_0 \Vert, \nonumber
\end{align}
and let $B_r, B_{r'} \subset \R^n$ be closed norm-balls centered at the origin of radius $r$ and $r'$, respectively. By assumption, we have that $\bs{\Phi}^{-1}$ and $\bs{\Phi}$ are globally Lipschitz continuous over, $\bs{\Phi}(\mathcal{X})\cap B_r$ and $\mathcal{X}\cap B_{r'}$ with Lipschitz constants $L_{\bs{\Phi}^{-1}}, L_{\bs{\Phi}} \in \R_{++}$, respectively. Given an arbitrary $R'\in\R_{++}$, pick $R<\min\{\overline{R},R'/L_{\bs{\Phi}^{-1}}\}$. Let $\mb{x}_0\in N'$ and let $\bs{\xi}_0 = \bs{\Phi}(\mb{x}_0)\in \bs{\Phi}(N') \subseteq N$. By the $(\beta,N)$-practical stability of the exact normal family, corresponding to $R$, there exists an $h^*\in I$ such that for any sample period $h\in(0,h^*)$, the recursion $\bs{\xi}_{k + 1} = \mb{F}_h^{e, \bs{\xi}}(\bs{\xi}_k, \mb{k}_h(\bs{\xi}_k))\in\bs{\Phi}(\mathcal{X})$ satisfies:
\begin{equation}
    \Vert \bs{\xi}_k \Vert \leq \beta(\Vert \bs{\xi}_0 \Vert,  kh)+R, \nonumber
\end{equation}
and implies $\bs{\xi}_k \in \bs{\Phi}(\mathcal{X})\cap B_r$ for all $k \in \Z_+$. Letting $\mb{x}_k = \bs{\Phi}^{-1}(\bs{\xi}_k)$ and $\mb{x}_{k + 1} = \mb{F}_h^{e, \mb{x}}(\mb{x}_k, \mb{k}_h(\bs{\Phi}(\mb{x}_k)))$ for all $k \in \Z_+$, note that $\mb{x}_k = \bs{\Phi}^{-1}(\bs{\xi}_k)$. It follows that:
\begin{align}
    &\Vert \mb{x}_k \Vert = \Vert \bs{\Phi}^{-1}(\bs{\xi}_k) - \bs{\Phi}^{-1}(\mb{0}_n) \Vert \leq L_{\bs{\Phi}^{-1}} \Vert \bs{\xi}_k - \mb{0}_n \Vert \nonumber\\
    &~\leq L_{\bs{\Phi}^{-1}} (\beta( \Vert \bs{\xi}_0 \Vert, kh ) + R) \nonumber\\
    &~= L_{\bs{\Phi}^{-1}}(\beta(\Vert \bs{\Phi}(\mb{x}_0) - \bs{\Phi}(\mb{0}_n) \Vert, kh) + R) \nonumber\\
    &~\leq L_{\bs{\Phi}^{-1}} (\beta(L_{\bs{\Phi}}\Vert \mb{x}_0 - \mb{0}_n \Vert, kh) + R) < \beta'(\Vert\mb{x}_0\Vert, kh) + R', \nonumber
\end{align}
for all $k \in \Z_+$, where $\beta'(r, s) = L_{\bs{\Phi}^{-1}}\beta(L_{\bs{\Phi}}r, s)$ for all $r, s \in \R_+$. As $\mb{x}_0 \in N'$ and $R'$ were arbitrary, we have that the exact state family is $(\beta',N')$-practically stable.
\end{proof}

\begin{proof}[Proof of Lemma \ref{lem:exactpracstab}]
Consider the open set $N \subseteq \bs{\Phi}(\mathcal{X})$ and the functions $\alpha_1, \alpha_2 \in \mathcal{K}_\infty$ and $\alpha_3\in\K$ as specified by Definition \ref{def:as-by-equi-lip-lyap}. Let $K \subset N$ be a compact set with $\mb{0}_n \in \mathrm{int}(K)$. By one-step consistency and asymptotic stability by equi-Lipschitz Lyapunov functions, there exist a $\rho \in \mathcal{K}_\infty$, a Lipschitz constant $M \in \R_+$, and an $h_0^* \in I$ such that for all $h \in (0, h_0^*)$, \eqref{eqn:one-step-cons}, \eqref{eqn:xibounds}, \eqref{eqn:decrease-bound}, and \eqref{eqn:lipschitzV} hold for all $\bs{\xi}, \bs{\xi}_1, \dots, \bs{\xi}_4 \in K$. There exists a radius $R \in \R_{++}$ such that the closed norm-ball around the origin of radius $R$ is contained in $K$. We modify the claim of \cite[eq. 37]{nevsic1999sufficient} for the local setting in this work as follows:
\begin{claim}
For any $d, D\in\R_{++}$ with $D \leq \alpha_2^{-1}(\alpha_1(\frac{R}{2}))$ and $d \leq 2 \alpha_2(R)$, there exists an $h^*\in(0, h_0^*)$ such that for every $\bs{\xi} \in \bs{\Phi}(\mathcal{X})$ and $h\in(0,h^*)$, if $\Vert \bs{\xi} \Vert \leq D$ and $\max{\{ V_h(\mb{F}_h^{e,\bs{\xi}}(\bs{\xi},\mb{k}_h(\bs{\xi}))), V_h(\bs{\xi}) \}} \geq d$, then:
\begin{equation}
   V_h(\mb{F}_h^{e,\bs{\xi}}(\bs{\xi},\mb{k}_h(\bs{\xi})))-V_h(\bs{\xi})\leq-\frac{h}{2}\alpha_3(\Vert\bs{\xi}\Vert).
\end{equation}
\end{claim}
In the language of \cite{nevsic1999sufficient}, the restrictions on $d$ and $D$ imply $\Delta \leq R$, and the proof follows by replacing the sets $\mathscr{X}$ and $\mathscr{X}_1$ with $K$, the constant $M$ with the Lipschitz constant $M$ given above, and the constants $T_1^*$ and $T_2^*$ with $h_0^*$. Letting $U\subset K$ be the open ball of radius $\alpha_2^{-1}(\alpha_1(\alpha_2^{-1}(\alpha_1(\frac{R}{2}))))$, the modified claim may be used to prove the existence of a $\beta\in\K\mathcal{L}_\infty$ such that the family $\{(\mb{k}_h,\mb{F}_h^{e,\bs{\xi}}) ~|~ h \in I\}$ is $(\beta, N')$-practically stable for any open set $N'\subseteq U$ containing the origin by following the proof of Theorem 2 in \cite{nevsic1999sufficient}.
\end{proof}
\bibliographystyle{IEEEtran} 
\bibliography{main}

\end{document}